\documentclass[11pt]{article}

\newcommand{\version}{July 12, 2015}

\usepackage{amsthm,amsfonts,amsmath, amscd}
\usepackage{amssymb,amsmath, dsfont}
\usepackage{graphicx}

\swapnumbers
\theoremstyle{plain}

\newtheorem{thm}{THEOREM}[section]
\newtheorem{lm}[thm]{LEMMA}
\newtheorem{cl}[thm]{COROLLARY}

\theoremstyle{definition}

\theoremstyle{definition}
\newtheorem{remark}[thm]{Remark}
\newcommand{\upchi}{\raise1pt\hbox{$\chi$}}
\newcommand{\R}{{\mathord{\mathbb R}}}

\newcommand{\C}{{\mathord{\mathbb C}}}

\newcommand{\id}{{\mathds 1}}

\newcommand{\tr}{{\rm Tr}\ }

\numberwithin{equation}{section}
\def\Pn{\mathcal{P}_n}

\begin{document}


\def\tr{{\rm Tr}}

\title{Some Operator and Trace Function Convexity Theorems}
\author{\vspace{5pt} Eric A. Carlen$^1$, Rupert L. Frank$^2$ and
Elliott H. Lieb$^3$ \\
\vspace{5pt}\small{$1.$ Department of Mathematics, Hill Center,}\\[-6pt]
\small{Rutgers University, 110 Frelinghuysen Road,
Piscataway NJ 08854-8019}\\
\vspace{5pt}\small{$2.$ Department of Mathematics,}\\[-6pt]
\small{Caltech,
Pasadena, CA 91125}\\
\vspace{5pt}\small{$3.$ Departments of Mathematics and Physics, Jadwin
Hall,} \\[-6pt]
\small{Princeton University, Washington Road, Princeton, NJ
  08544}\\
 }
\date{\version}
\maketitle

\footnotetext                                                                         
[1]{Work partially
supported by U.S. National Science Foundation
grant DMS-1201354.    }
\footnotetext
[2]{Work partially
supported by U.S. National Science Foundation
grants PHY-1347399 and DMS-1363432.    }                                                          
\footnotetext
[3]{Work partially
supported by U.S. National Science Foundation
grant PHY-1265118.\\
\copyright\, 2015 by the authors. This paper may be reproduced, in its
entirety, for non-commercial purposes.}

\begin{abstract}
We consider trace functions $(A,B)\mapsto \tr[ (A^{q/2}B^p
A^{q/2})^s]$ where $A$ and $B$ are positive $n\times n$
matrices and ask when these functions are convex or concave. We also
consider operator convexity/concavity  of $A^{q/2}B^p A^{q/2}$ and
convexity/concavity of
the closely related trace functional $\tr[ A^{q/2}B^p
A^{q/2} C^r]$.  The concavity questions are completely
resolved, thereby settling cases left open by Hiai; the convexity
questions are settled in many cases. As a consequence, the  Audenaert--Datta
R\'enyi entropy conjectures are proved for some cases. 
\end{abstract}

\medskip
\leftline{\footnotesize{\qquad Mathematics subject classification numbers: 47A63, 94A17, 15A99}}
\leftline{\footnotesize{\qquad Key Words: Operator Convexity, Operator
Concavity, Trace inequality, R\'enyi Entropy}}

\section{Introduction}

Let $\Pn$ denote the set of $n\times n$ positive definite matrices.  For $p,q,s\in \R$, define
\begin{equation}\label{phidef}
\Phi_{p,q,s}(A,B) = \tr[ (A^{q/2}B^p A^{q/2})^s]\ .
\end{equation}
We are mainly interested in the convexity or concavity of the map 
$(A,B)\mapsto
\Phi_{p,q,s}(A,B)$, but we are also interested in the {\it operator}
convexity/concavity of $A^{q/2}B^p A^{q/2}$.
When any of $p$, $q$ or $s$
is zero, the question of convexity is trivial, and we exclude these cases.

Given any $n\times n$ matrix $K$, and with $p$, $q$, $s$ as above, define 
\begin{equation}\label{psidef}
\Psi_{K,p,q,s}(A,B)=   \tr[ (A^{q/2}K^* B^p  K A^{q/2})^s]\ ,
\end{equation}
and note that
\begin{equation}\label{psiphi}
\Phi_{p,q,s}(A,B) = \Psi_{\id,p,q,s}(A,B)  \ .
\end{equation}

The main question to be addressed here is this: {\em For which non-zero
values of $p$, $q$ and $s$ is  $\Psi_{K,p,q,s}(A,B)$
jointly convex or jointly concave on $\Pn\times \Pn$ for all $n$ and all
$K$?}

We begin with several simple reductions.  Since invertible $K$ are dense, it suffices to consider all invertible operators $K$. Then, for $K$ invertible,
$$\Psi_{K,p,q,s}(A,B) = \Psi_{(K^*)^{-1},-p,-q,-s}(A,B)\ ,$$
and therefore it is no loss of generality to assume that $s>0$. We
always make this assumption in what follows.

Next,  the convexity/concavity properties of $\Psi_{K,p,q,s}(A,B)$ are a
consequence of those of $\Phi_{p,q,s}(A,B)$, and hence it suffices to study
the special case $K=\id$. In fact, more is true as stated in the following
Lemma \ref{equiv}.
These equivalences may be useful in other contexts. (For $s=1$ the
equivalence of (1) and (4) is in \cite{Lieb73} and
the equivalence of (1) and
(3) is in \cite{Carlen08}; the arguments in those papers extend to  all
$s$,  but we repeat them here for completeness.)

\begin{lm}[\bf Equivalent formulations] \label{equiv}
The following statements are equivalent for fixed $p,q, s$.

\smallskip

 {\rm(1)} The map $(A,B) \mapsto \Psi_{K,p,q,s}(A,B)$ is convex
for all $K$ and all $n$.
\smallskip

{\rm(2)} The map $\!(A,B)\! \mapsto \!\Psi_{\!K,p,q,s}(\!A,\!B)\!$ is convex
for all unitary $\!K\!\!$ and all $\!n$.
\smallskip

{\rm(3)} The map $(A,B) \mapsto \Psi_{\id ,p,q,s}(A,B) =\Phi_{p,q,s}(A,B)$
is
convex for all $n$. 
\smallskip

{\rm(4)} The map $ A\mapsto \Psi_{K ,p,q,s}(A,A) $ is
convex for all $K$ and all $n$.
\smallskip

{\rm(5)} The map $ A\mapsto \Psi_{K ,p,q,s}(A,A) $ is
convex for all unitary $K$ and all $n$.
\smallskip

\noindent The same is true if convex is replaced by concave in all
statements.

\end{lm}

\begin{proof} Trivially, (1) implies the other four items.
\smallskip

When $K$ is unitary, $K^*A^q K = (K^*AK)^q$, and hence (3) implies (2) (even for each fixed $n$).
By taking $K=\id$, (2) implies (3) (again for each fixed $n$). 
\smallskip

Next we show that (2) implies (1), whence (1), (2) and (3) are equivalent. 
We may suppose, without loss of generality that $K$ is a
contraction. Let $K = W|K|$ 
be its polar decomposition. Then 
$$\mathcal{U} = \left[\begin{array}{cc} K & W \sqrt{\id-|K|^2}\\
-W\sqrt{\id-|K|^2} & K \end{array}\right]$$
is unitary.  We consider the case $q < 0$ first. 
For arbitrary $t>0$, let 
$$\mathcal{A}_t = \left[\begin{array}{cc} A & 0\\ 0 &
t\id\end{array}\right]\ ,\  \mathcal{B} = \left[\begin{array}{cc} B & 0\\ 0 &
0\end{array}\right] \,.$$
Then
$$\left[\begin{array}{cc} A^{q/2}K^* B^p  K A^{q/2}  & 0\\ 0 & 0
\end{array}\right] =  \lim_{t\to\infty}\mathcal{A}_t^{q/2}\mathcal{U} ^*
\mathcal{B}^p  \mathcal{U}  \mathcal{A}_t^{q/2} \, .$$
Thus, recalling that we always assume $s>0$, 
$$\tr[ (A^{q/2}K^* B^p  K A^{q/2} )^s] = \lim_{t\to\infty}
\Psi_{\mathcal U,p,q,s}(\mathcal{A}_t,\mathcal{B})\ .$$
Thus, (2) with $2n$ implies (1) with $n$. 
The case $q>0$ is treated analogously, letting $t\to 0$. 
\smallskip

Trivially, (4) implies (5). To show that (5) (with $2n$) implies
(3) (with $n$), thereby completing the loop, replace $A$ in (5) by
$\left[\begin{array}{cc} A & 0\\ 0 &
B\end{array}\right]$, and replace $K$ by the unitary
$\left[\begin{array}{cc} 0 & \id \\ \id  &
0\end{array}\right]$.
\end{proof}


\section{Known results and our extension of them}\label{sec:known}

Hiai has proved in \cite{Hiai} that if $p$, $q$  are both non-zero, and $s>0$,
and $\Phi_{p,q,s}$ is jointly {\it convex} in $A$ and $B$, then,
\textit{necessarily}, one of the following conditions holds:

\smallskip
\noindent{\it (1.)} \ $1\leq p \leq 2$ and $-1 \leq q < 0$ and $s\geq
1/(p+q)$, or the same with $p$ and $q$ interchanged.

\smallskip
\noindent{\it (2.)} \ $-1\leq p, q < 0$  and $s>0$.

In the special case $s=1$, condition {\it (1.)} was proved  to be
sufficient in 
\cite[Corollary 6.3]{Ando}, and condition {\it (2.)} was proved to be
sufficient 
in \cite[Theorem 8]{Lieb73}; see also
\cite{Bekjan} for $s=1$ and one of  $p,q$ negative. 
Hiai \cite{Hiai} has also proved that $\Phi_{p,q,s}$ is jointly convex in case
$-1\leq p,q < 0$ and $1/2 \leq s \leq -1/(p+q)$.\footnote{After this work was submitted, Hiai posted the preprint arXiv:1507.00853 in which he extended our method to prove joint convexity under condition (2.).} Our main focus is on {\it (1.)}. The joint convexity in this case is known \cite{Frank13} when $s=1/(p+q)$, $p=1$ and $-1\leq q < 0$, and of course, with $p$ and $q$ interchanged.

Concerning {\it concavity,} Hiai has shown \cite{Hiai} that if $p$, $q$  are
both non-zero, and $s>0$,
and $\Phi_{p,q,s}$ is jointly concave  in $A$ and $B$, then, \emph{necessarily}, the following condition holds:
\smallskip

\noindent
{\it (3.)}\  $0< p,q \leq 1$ and  $0<s \leq 1/(p+q)$.

In the special case $s=1$, this condition was proved to be
sufficient in \cite[Theorem 1]{Lieb73}; Hiai \cite{Hiai} showed sufficiency for
${1}/{2}\leq s \leq 1/(p+q)$. 

\medskip

Our contribution to the subject is to fill in parts of the table of
sufficient/necessary conditions in the following manner.
We were motivated
in this endeavor by a recent paper of Audenaert and Datta \cite{Audenaert},
(and Datta's Warwick lecture on it) and we prove some of their
conjectures.

All the results mentioned above refer to trace inequalities. There are some
{\it operator} convexity/concavity inequalities to be considered as well,
and we will present some in the following. 

As far as convexity of $\Phi_{p,q,s}$ is concerned we can summarize our
results as follows. We are concerned with the region $p\in [1,2] $, $q\in
[-1,0)$ and $s\geq 1/(p+q)$. (Clearly, $s$ cannot be smaller than $1/(p+q)$
by homogeneity.) We prove joint convexity for $s \geq \min\left\{ \frac{1}{p-1}\ , \ \frac{1}{1+q}\right\}$ (Thm. \ref{larges}). Moreover, we prove joint convexity for $p=1$ and $p=2$ in the optimal range $s\geq 1/(p+q)$ (Thm. \ref{p2}).

For $p\in (1,2)$, $q\in [-1,0)$, the missing regions, where we believe joint convexity also holds, is $1/(p+q) \leq s <1$ and $1<s< \min\left\{
\frac{1}{p-1}\ , \ \frac{1}{1+q}\right\}$. (Ando's theorem \cite{Ando} covers the cases $1/(p+q)\leq s=1$.)

On the other hand, our results completely close the gap between 
necessary and sufficient conditions for  concavity to hold.
The trace function $\Phi_{p,q,s}$ is jointly concave if and only if
$0<p,q \leq 1$ and $0\leq s \leq 1/(p+q)$ (Thm. \ref{conc}). This completes
Hiai's results discussed above. 

As for joint {\it operator} convexity, we prove it for $(A,B)\mapsto BA^qB$ if $-1\leq q<0$,
and show that it does {\it not hold} for $(A,B)\mapsto
B^{p/2}A^qB^{p/2}$ for any $p<2$ (Thm. \ref{opcon}). (Note that it cannot
hold for $p>2$ since
$B\mapsto B^p$ is not operator convex when $p>2$.)

\section{Joint operator convexity}

We investigate  operator convexity and concavity  of certain functions on $\Pn\times \Pn$. It is well known \cite{Kiefer, Ruskai} that
\begin{equation}\label{opcon1}
(A,B) \mapsto AB^{-1}A\ 
\end{equation}
is jointly convex. In the scalar case $(n=1)$, $f(a,b) = a^qb^p$ is jointly convex on $(0,\infty)\times(0,\infty)$ if and only if  $p\geq 1$, $q\leq 0$ and $p+q \geq 1$, or 
$q\geq 1$, $p\leq 0$ and $p+q \geq 1$, or $p,q \leq 0$. It is jointly concave if and only if $0 \leq p,q \leq 1$ and $p+q \leq 1$. It is natural to ask for which powers $p$ and $q$
\begin{equation}\label{pow}
(A,B) \mapsto A^{q/2}B^pA^{q/2}
\end{equation}
is jointly operator convex or concave.

This question is closely related to the question: For which values of
$p,q,r$ is 
\begin{equation} \label{pow2}
(A,B,C) \mapsto \tr A^{q/2}B^pA^{q/2}C^r
\end{equation}
jointly convex or concave in the positive
operators $A,B,C$?  

\begin{lm}\label{relation} When the function in (\ref{pow2}) is convex (or
concave) for some choice of $p$, $q$ and $r$ all non-zero, then 
the function in (\ref{pow}) is operator convex (or concave) for the same $p$
and $q$.
\end{lm}

\begin{proof}  When $r$ is positve, simply take $C$ to be any rank-one projection.  
 When $r$ is negative, let $P$ be any rank-one projection, $t>0$. Take $C$ to be
$P+tP^\perp$, so that $C^r =  P+t^r P^\perp$ and let $t$ tend to $\infty$. 
\end{proof}

Thus, the operator
convexity/concavity  of the operator-valued function in \eqref{pow} is a 
consequence of the seemingly weaker tracial convexity/concavity of 
\eqref{pow2}. In short, \eqref{pow2} is {\em stronger} than \eqref{pow} for
the same values of $p,q$. The value of $r$ is irrelevant as long as it is
not zero, and the implication does not even require convexity/concavity in
$C$, only joint convexity/concavity in $A$ and $B$. 

When $p,r<0$, and $-1\leq p+r <0$,
then the map $(A,B,C) \mapsto \tr A^{}B^pA^{*}C^r$ is
jointly convex for $B,C$ positive and $A$ arbitrary. This was proved in
\cite[Corollary 2.1]{Lieb73}. (This {\it triple convexity theorem} is deeper
than the double convexity theorem \cite[Theorem 8]{Lieb73} referred to in
the previous section because it uses \cite[Theorem 2]{Lieb73} in an
essential way.)  By 
restricting ourselves to $A$ positive and 
taking $q=2$ this function of $A,B,C$ reduces to \eqref{pow2}.

By Lemma~\ref{relation}, the function \eqref{pow}
is jointly convex when $q=2$ and $-1\leq p <0$. Our main result in this
section is that there are no other cases in which this
operator-valued function  is either
convex or concave\,!

\begin{thm}\label{opcon}
Let $p,q\in\R\setminus\{0\}$ and consider the map
\begin{equation}
\label{pow1}
(A,B) \mapsto A^{q/2} B^{p}A^{q/2}
\end{equation}
from $\Pn \times \Pn$ to $\Pn$ for some fixed $n\geq 2$.

\smallskip

\noindent{\it (1.)} The map \eqref{pow1} is jointly operator convex if and only if 
$q=2$ and $-1\leq p<0$.

\smallskip

\noindent{\it  (2.)} The  map \eqref{pow1} is {\rm not} jointly operator concave.
\end{thm}

\begin{cl}
Let $p,q\in\R\setminus\{0\}$.
The function $
(A,B,C) \mapsto \tr A^{q/2}B^pA^{q/2}C^r$
is never concave,  and it is convex if and only if $q=2$, $p,r<0$ and 
$-1\leq p+r <0$.
\end{cl}

\begin{proof} 
 By Lemma~\ref{relation}, any triple
convexity/concavity would imply the corresponding operator
convexity/concavity, which is ruled out by the previous Theorem \ref{opcon},
except when 
  $q=2$, $p,r<0$ and $-1\leq p+r <0$. In this case convexity is provided by
\cite[Corollary 2.1]{Lieb73}. 
\end{proof}

Our counterexamples to operator convexity and concavity given in  Theorem
\ref{opcon} will be based on the following lemma.

\begin{lm}\label{counter}
Let $r\in(-\infty,0)\cup(0,1)$, let $Y\geq 0$ be rank one and $n\geq 2$. 
Then the map $X\mapsto X^r Y X^r$ from $\Pn$ to $\Pn$ is {\rm not} operator
convex. 
\end{lm}

\begin{proof}[Proof of Lemma \ref{counter}]
First assume that $r\in(0,1/2)$. Then for any non-trivial $Y\geq 0$ (not
necessarily rank one) the map $X\mapsto X^r Y X^r$ from $\Pn$ to $\Pn$ is
not operator convex. This follows simply from the fact that the map
$x\mapsto x^{2r} Y$ from $(0,\infty)$ to $\Pn$ is not operator convex for
$0<r<1/2$. It is, in fact, strictly concave in this region.

Now let $r\in(-\infty,0)$. (The proof actually also works  for
$r\in(0,1/2)$, which is hardly surprising in light of the concavity
mentioned above.)  Clearly, we may assume $n=2$. Let $Y = |v\rangle\langle
v|$. If the convexity were
true, then for all $X_1,X_2\in \mathcal{P}_2$, with $X = (X_1+X_2)/2$,  we
would have
\begin{equation}\label{contra}
X^r |v\rangle \langle v |X^r \leq \tfrac12 X_1^r |v\rangle
\langle v |X_1^r + \tfrac12 X_2^r |v\rangle \langle v |X_2^r\ .
\end{equation}
Without loss of generality, let $|v\rangle = (1,\, 1)$. 
If we take $X_1 = \left[\begin{array}{cc} 2 & 0\\ 0 &
2\end{array}\right]$ and $X_2 = t \left[\begin{array}{cc} 2 & 0\\ 0 &
4\end{array}\right]$, with $t>0$, then (\ref{contra}) becomes
\begin{multline} \label{contra2}
\left[\begin{array}{cc} (1+t)^{2r}  & (1+t)^{r}(1+2t)^{r}  \\ 
(1+t)^{r}(1+2t)^{r} & (1+2t)^{2r} \end{array}\right]  \\    
\leq 2^{2r-1}\left[\begin{array}{cc} 1 & 1\\ 1 & 1\end{array}\right]
+ t^{2r}2^{2r-1}\left[\begin{array}{cc} 1 & 2^r\\ 2^r & 2^{2r}\end{array}\right]
\ .\end{multline}
The vector $|w\rangle = (2^r, -1)$ is in the null space of the second matrix
on the right in (\ref{contra2}), and taking the trace of 
both sides against $|w\rangle\langle w|$ yields
$$\left\langle w, \left[\begin{array}{cc} (1+t)^{2r}  & (1+t)^{r}(1+2t)^{r}  \\ (1+t)^{r}(1+2t)^{r} & (1+2t)^{2r} \end{array}\right] w\right\rangle \leq 
2^{2r-1} (2^r -1)^2\ ,$$
which, in the limit $t\to0$, becomes $(2^r -1)^2 \leq 2^{2r-1} (2^r -1)^2$,
so that for $r\neq 0$, we would have $1 \leq 2^{2r-1}$. This is false for
all
$r < 1/2$, which shows that  (\ref{contra}) leads to a contradiction for
nonzero $r\in(-\infty,0)\cup(0,1/2)$.

Our proof for $1/2\leq r <1$ is different; this proof actually works 
in the range $0<r<1$. Let $|v\rangle$ be a unit vector in $\C^n$. Then we
will show that there is another vector $|w\rangle$ in $\C^n$ such that $$X
\mapsto  | \langle w| X^{r} |v\rangle|^2$$ is not convex. Again, we may
assume that $n=2$ and that $|v\rangle = (0,1)$. Take
$$X_1 = \left[\begin{array}{cc} 2 & 2\\ 2 & 2\end{array}\right]\quad{\rm
and}\quad   X_2 =  \left[\begin{array}{cc} 2 & 0\\ 0 & 0\end{array}\right]\
.$$
Let $|w\rangle= (1,-1)$, so that $X_1^{r}|w \rangle=
0$ and $X_2^{r}|v\rangle = 0$. Evidently,
$$\tfrac12  | \langle w| X_1^{r} |v\rangle|^2 + \tfrac12  | \langle
w| X_2^{r} |v\rangle|^2 = 0\ .$$
However, the eigenvalues of $X = \tfrac12(X_1+X_2)$ are easily computed to
be $\lambda_\pm = (3\pm \sqrt{5})/2$, and then a further simple computation
yields
\begin{eqnarray}
\langle w| X^{r} |v\rangle  = \frac{1}{\sqrt{5}}(\lambda_+^{r-1} - 
\lambda_-^{r-1})\ ,\nonumber
\end{eqnarray}
and this is strictly negative for all $0<r<1$.
\end{proof}

\begin{proof}[Proof of Theorem \ref{opcon}]
 As explained above, 
the convexity assertion in {\it (1.)}
is a consequence of \cite[Corollary 2.1]{Lieb73}. Our goal now is
to prove that there are no other cases of convexity or concavity. 

A number of exponents can be excluded by considering the scalar 
case.
Moreover, since $X\mapsto X^r$
is operator convex on $\Pn$ if and only if $r \in [-1,0]\cup [1,2]$, 
and is operator concave on  $\Pn$ if and only if $r \in [0,1]$, the only
cases
in which   convexity  cannot be immediately ruled out are  
$p\in [1,2]$, $q\in [-1,0]$ and $p+q \geq 1$ (or the same with $p$ and $q$
interchanged). 
Likewise,  the only cases of in which concavity cannot be immediately 
ruled out are $p,q \in [0,1]$, $p+q \leq 1$.

For part {\it (1.)} it remains for us to show that
\eqref{pow1} is not jointly operator convex in the following three cases,

\smallskip
\noindent{\it (a)} \ $p\in[-1,0)$, $q\in [1,2)$ and $p+q\geq 1$.

\smallskip
\noindent{\it (b)} \ $p\in[1,2]$, $q\in [-1,0)$ and $p+q\geq 1$.

\smallskip
\noindent{\it (c)} \ $p\in(-1,0)$ and $p+q\geq -1$.

\smallskip

Let us prove failure of convexity in case ${\it (a)}$. Let $|v\rangle$ be any unit
vector in $\C^n$. Let $P$ be the orthogonal projection
onto the span of $v$, and let $P^\perp$ denote the complementary
projection. 
Fix $t> 0$, and define $B_t = P+tP^\perp$. Then
$B_t^p = P + t^p P^\perp$. If convexity would hold, then for any $|w\rangle$ the map
$A \mapsto \langle w| A^{q/2}B_t^p A^{q/2}|w\rangle$ would be convex. Since $\lim_{t\to\infty}B_t^p = |v\rangle \langle v|$, 
and since limits of convex functions are convex, it would follow that
$A \mapsto |\langle v |A^{q/2}|w\rangle |^2$ would be convex on
$\Pn$ for any $|w\rangle$. This contradicts Lemma \ref{counter} with $r=q/2\in[1/2,1)$. The proof for 
${\it (c)}$ is almost exactly the same, except one uses Lemma \ref{counter} with $r =q/2 < 0$.

The proof in case ${\it (b)}$ is similar. Again, we let $|v\rangle$ be a
unit 
vector in $\C^n$ and set $B=|v\rangle\langle v|$. Then
$B^p=|v\rangle\langle v|$ and, if convexity would hold, then for any
$|w\rangle$ the map $A\mapsto |\langle v |A^{q/2}|w\rangle |^2$ would be
convex on $\Pn$. This contradicts Lemma \ref{counter} with
$r=q/2\in[-1/2,0)$.

Finally, we prove {\it (2.)}, the failure of concavity. According to the
discussion above, it remains for us to show that \eqref{pow1} is not jointly operator concave for $p,q\in (0,1]$ and $p+q\leq 1$.
Suppose $(A,B)
\mapsto A^{q/2}B^{p}A^{q/2}$ were concave for some $p,q$ in this range. Then for all non-negative $A$ and
$B$ we would have
\begin{eqnarray}
\tfrac12 A^{q/2}B^pA^{q/2} + \tfrac12 B^{q/2}A^{p}B^{q/2} &\leq & \left(\frac{A+B}{2}\right)^{q/2} \left(\frac{B+A}{2}\right)^p \left(\frac{A+B}{2}\right)^{q/2}\nonumber\\
&=& 2^{-p-q}(A+B)^{p+q}\ .\nonumber
\end{eqnarray}
Suppose that $A$ has a non-trivial null space (here we use the assumption $n\geq 2$), and $|v\rangle$ is a unit
vector with $A|v\rangle =0$. By Jensen's inequality, since $p+q\leq 1$,
$$\langle v|(A+B)^{p+q}| v\rangle  \leq \langle v| (A+B) |v\rangle^{p+q}
= \langle v| B |v\rangle ^{p+q}\ .$$
Thus we would have 
$$\langle v| B^{q/2}A^{p}B^{q/2}|v\rangle \leq   2^{1-p-q}\langle v|B 
|v\rangle ^{p+q}\ .$$
The left side is homogeneous of degree $q$ in $B$, while the right side is
homogeneous of degree $p+q$, and hence the inequality cannot be generally
valid. (The positivity of the powers is essential here; the argument of course cannot be adapted to yield a counterexample to the convexity proved in the first part of the theorem.)
\end{proof}

\begin{remark}
There is another way to prove the convexity in   (\ref{pow1}) for 
$q =2 $ and $ -1 \leq p<0$.
For $p=-1$ one can use the Schwarz type inequality in 
\cite{Ruskai, Kiefer}. (This
inequality, however, is actually weaker than the triple convexity
inequality \cite[Corollary 2.1]{Lieb73} that we used in
the proof of Theorem 3.2.)
For $-1<p<0 $ one can use the integral representation 
$B^p = C_p\int_0^\infty
(B+t)^{-1} t^{p} {\rm d}t$ with $C_p>0$
to reduce matters to the case $p=-1$. Indeed, one can replace
$B^{p}$ by any Herglotz function $\int_{t\geq 0} (B+t)^{-1} {\rm d}\mu(t)$
with $\mu >0$.
\end{remark}


\section{Convexity of $\Phi_{p,q,s}(A,B)$}

In this section we prove, among other things, two cases  of a conjecture of Audenaert and Datta \cite{Audenaert}. Much of our analysis is based on the formulas
\begin{equation}\label{var0}
 \tr[X^s] = s\sup_{Z \geq 0}\left\{\tr[XZ^{1-1/s}] + \left(\tfrac{1}{s} -1\right)\tr[Z]\right\} \qquad\text{if}\ s>1
 \end{equation}
and
\begin{equation}\label{var00}
 \tr[X^s] = s\inf_{Z > 0}\left\{\tr[XZ^{1-1/s}] + \left(\tfrac{1}{s} -1\right)\tr[Z]\right\} \qquad\text{if}\ 0 < s<1 \,;
 \end{equation}
see \cite[Lemma 2.2]{Carlen08}. These formulas have already played an important role in our previous works \cite{Carlen08} and \cite{Frank13}.

\begin{thm}\label{larges}
When $p\in [1,2]$, $q \in [-1,0)$, $\Phi_{p,q,s}(A,B)$ is jointly convex for all 
$$s \geq \min\left\{ \frac{1}{p-1}\ , \ \frac{1}{1+q}\right\}\ . $$
\end{thm}

Here we set $\frac{1}{p-1}=+\infty$ for $p=1$ and $\frac{1}{1+q}=+\infty$ for $q=-1$. Thus, the theorem implies that, in particular, for $p=1$, $\Phi_{1,q,s}(A,B)$ is jointly convex in the optimal range $q\in[-1,0)$ and $s\geq\frac{1}{1+q}$. An optimal result for $p=2$ will be proved in Theorem \ref{p2}. As discussed in Section \ref{sec:known}, for $p\in (1,2)$, $q \in [-1,0)$, the region where convexity is not settled is $1/(p+q)\leq s<1$ and $1<s<\min\left\{ \frac{1}{p-1}\ , \ \frac{1}{1+q}\right\}$.

\begin{proof}
First, we prove convexity if $s\geq 1/(1+q)$. Since this implies $s>1$, we have by \eqref{var0},
$$\Phi_{p,q,s}(A,B) =  s\sup_{Z \geq 0}\left\{\tr[A^{q/2}B^pA^{q/2}Z^{1-1/s}] + \left(\tfrac{1}{s} -1\right)\tr[Z]\right\}\ .$$
Now define $D^2 =  A^{q/2}Z^{(s-1)/s}A^{q/2}$ and note that $Z= (A^{-q/2} D^2 A^{-q/2})^{s/(s-1)}$ to write
\begin{equation}
\label{eq:largesproof}
\Phi_{p,q,s}(A,B) = s\sup_{D \geq 0}\left\{\tr[D B^pD] + \left(\tfrac{1}{s} -1\right)\tr[(D A^{-q} D)^{s/(s-1)}]\right\}\ .
\end{equation}
For $1\leq p\leq 2$, the map $B\mapsto B^p$ is operator convex and therefore $B\mapsto \tr[D B^pD]$ is convex. Moreover, by Hiai's extension of Epstein's Theorem \cite[Thm. 4.1]{Hiai} the map $A\mapsto \tr[(D A^{-q} D)^{s/(s-1)}]$ is concave as long as $s/(s-1) \leq -1/q$, which is the same as $s\geq 1/(1+q)$. Thus, \eqref{eq:largesproof} represents $\Phi_{p,q,s}(A,B)$ as a supremum of jointly convex functions and so $\Phi_{p,q,s}(A,B)$ is jointly convex for $s\geq 1/(1+q)$. This proves the first part of the theorem.

We now prove convexity if $s\geq 1/(p-1)$. Let us first consider the case $p=2$ and $s=1$, where $\Phi_{2,q,1}(A,B)=\tr[A^{q/2}B^2 A^{q/2}]=\tr[B A^q B]$. For $-1\leq q<0$, the map $(A,B)\mapsto B A^q B$ is operator convex by Theorem~\ref{opcon} and therefore $(A,B)\mapsto\tr[B A^q B]$ is convex, as claimed. We now assume that $s>1$ (and still $s\geq 1/(p-1)$). Then by \eqref{var0}, making use of $\tr[(A^{q/2}B^p A^{q/2})^s]  = \tr[(B^{p/2}A^{q}B^{p/2})^s]$,
$$
\Phi_{p,q,s}(A,B) =  s\sup_{Z \geq 0}\left\{\tr[B^{p/2}A^qB^{p/2}Z^{1-1/s}] + \left(\tfrac{1}{s} -1\right)\tr[Z]\right\}\ .
$$
Note that
$$
\tr[B^{p/2}A^qB^{p/2}Z^{1-1/s}] = \tr[BA^qB(B^{p/2-1}Z^{1-1/s}B^{p/2-1})]\ .
$$
Define $D^2 = B^{p/2-1}Z^{(s-1)/s}B^{p/2-1}$, so that $Z =   (B^{1-p/2} D^2 B^{1-p/2} )^{s/(s-1)}$. Then
\begin{align}
\label{var11}
\Phi_{p,q,s}(A,B) & =  s\sup_{D \geq 0}\left\{\tr[DBA^qBD] + \left(\tfrac{1}{s} -1\right) \tr[(B^{1-p/2} D^2 B^{1-p/2} )^{s/(s-1)}] \right\} \notag \\
& =  s\sup_{D \geq 0}\left\{\tr[DBA^qBD] + \left(\tfrac{1}{s} -1\right)\tr[(D B^{2-p} D )^{s/(s-1)}]\right\} \ .
\end{align}
Since $-1\leq q<0$, $(A,B)\mapsto BA^qB$ is operator convex by Theorem \ref{opcon}, so $(A,B)\mapsto\tr[DBA^qBD]$ is convex. By Hiai's extension of Epstein's Theorem \cite[Thm. 4.1]{Hiai}, $B \mapsto \tr[(D B^{2-p} D )^{s/(s-1)}]$ is concave as long as $s/(s-1) \leq 1/(2-p)$, which is the same as $s \geq 1/(p-1)$. Thus, \eqref{var11} represents $\Phi_{p,q,s}(A,B)$ as a supremum of jointly convex functions and so $\Phi_{p,q,s}(A,B)$ is jointly convex for $s\geq1/(p-1)$. This completes the proof.
\end{proof}

\begin{thm} \label{p2}
When $p=2$, $\Phi_{p,q,s}(A,B)$ is jointly convex for all $-1 \leq
q < 0$ and $s \geq1/(2+q)$. 
\end{thm}

This result yields the optimal range of convexity for $p=2$. It had been conjectured in \cite{Audenaert}
for $s = 1/(2+q)$. 

\begin{proof} 
The convexity for $s\geq 1$ follows from Theorem \ref{larges} and therefore we may assume that $1/(p+q)\leq s<1$. Then, making use of $\tr[(A^{q/2} B^2 A^{q/2})^s] = \tr[(B A^q B)^s]$,
 \begin{equation}\label{var000}
\Phi_{2,q,s}(A,B) = s\inf_{Z > 0}\left\{\tr[BA^qBZ^{1-1/s}] + \left(\tfrac{1}{s} -1\right)\tr[Z]\right\}\ .
 \end{equation}
The important distinction between this formula and formulas \eqref{eq:largesproof} and \eqref{var11} is the infimum in place of the supremum. Joint convexity in $A, B$ no longer suffices. Instead we need joint convexity in $A,B, Z$, with which we can apply \cite[Lemma 2.3]{Carlen08}. 
 
 Note that $1-1/s \leq 0$. By \cite[Corollary 2.1]{Lieb73}, $(A,B,Z) \mapsto
\tr[BA^qBZ^{1-1/s}]$ is jointly convex
 as long as $q+ 1-1/s \geq -1$, which means $s\geq 1/(2+q)$.  For such
$s$, the argument of the infimum in 
 (\ref{var000}) is jointly convex in $A$, $B$ and $Z$. 
 By \cite[Lemma 2.3]{Carlen08}, the infimum itself is jointly convex in $A$ and $B$. This proves the assertion for $ 1/(2+q) \leq s \leq 1$.
 \end{proof}
 
\begin{remark}
In the previous proof for the range $s\geq 1$ we referred to Theorem \ref{larges} which, in turn, was based on Hiai's extension of Epstein's theorem. For the case relevant for Theorem \ref{p2}, however, there is a more direct proof. Indeed, let $A_j,B_j\in \Pn$, $j=1,2$, and $\lambda \in (0,1)$ and set $A = \lambda A_1+ (1-\lambda)A_2$
 and $B= \lambda B_1+ (1-\lambda)B_2$. Then by Theorem \ref{opcon} for $-1\leq q<0$,
 $$BA^qB \leq \lambda B_1A_1^qB_1 + (1-\lambda) B_2A_2^qB_2\ .$$
 For all $s\geq 0$, $X\mapsto \tr[X^s]$ is monotone on $\Pn$. Hence, even for all $s\geq 0$, 
 $$\tr[(BA^qB)^s] \leq \tr[(\lambda B_1A_1^qB_1 + (1-\lambda) B_2A_2^qB_2)^s]\ .$$
 Finally, for $s\geq 1$, $X\mapsto \tr[X^s]$ is convex on $\Pn$. Therefore,
 $$
 \tr[(\lambda B_1A_1^qB_1 + (1-\lambda) B_2A_2^qB_2)^s] \leq \lambda \tr[(B_1A_1^qB_1)^s] + (1-\lambda) \tr[(B_1A_1^qB_1)^s]\ .$$
 This proves the convexity for $s\geq 1$ and $-1\leq q<0$. 
\end{remark}

The next result concerns the concavity of $\Phi_{p,q,s}(A,B)$.

\begin{thm}\label{conc}
The trace function $\Phi_{p,q,s}(A,B)$ is jointly concave if and only if
$0\leq p,q \leq 1$ and $0\leq s \leq 1/(p+q)$.
\end{thm}

\begin{proof}
The necessity of the condition is proved in \cite[Prop. 5.1]{Hiai} and  the
sufficiency for $1/2 \leq s \leq 1/(p+q)$ is proved in \cite[Thm. 2.1]{Hiai}. Our task
is to prove sufficiency in the case $0<s<1/2$. We write, using \eqref{var00},
\begin{align*}
\Phi_{p,q,s}(A,B) &= s \inf_{X>0} \tr\, \left\{ A^{q/2} B^p A^{q/2} X^{1-1/s}   
+(\tfrac{1}{s} -1)X \right\}\\
&= s \inf_{Y>0} \tr \left\{ B^p Y  +(\tfrac{1}{s}-1) \left( A^{q/2} Y^{-1} A^{q/2}
\right)^{s/(1-s)} \right\}\\
&= s \inf_{Y>0} \tr \left\{ B^p Y  +(\tfrac{1}{s}-1) \left( Y^{-1/2} A^q Y^{-1/2}
\right)^{s/(1-s)} \right\}\ .
\end{align*}
Since $0\leq p\leq 1$, $B\mapsto B^p$ is operator concave and so $B\mapsto\tr B^pY$ is concave. By the extension of Epstein's Theorem  proved in \cite[Theorem 4.1]{Hiai}, $A \mapsto \tr (Y^{-1/2} A^q Y^{-1/2})^{s/(1-s)}$ is concave if $s/(1-s)\leq 1/q$. This condition is satisfied since $s\leq 1/2\leq 1/(1+q)$. We conclude that $\Phi_{p,q,s}(A,B)$ as an infimum of concave functions is concave.
\end{proof}

We conclude with a corollary of Theorem \ref{p2}. For $\rho,\sigma\in\Pn$ and $\alpha,z>0$, we introduce the so-called $\alpha-z$-relative R\'enyi entropies
$$
D_{\alpha,z}(\rho||\sigma) = \frac1{\alpha-1} \ln \frac{\tr\left( \sigma^{(1-\alpha)/(2z)} \rho^{\alpha/z} \sigma^{(1-\alpha)/(2z)}\right)^z}{\tr\rho} \,.
$$
(For $\alpha=1$, a limit has to be taken.) These functionals appeared in
\cite[Sec. 3.3]{Jaksic} and were further studied  in
\cite{Audenaert}, where the question was raised
whether the $\alpha-z$-relative R\'enyi entropies are monotone under
completely positive, trace preserving maps. Currently this is known for
$0<\alpha\leq 1$ and $z\geq\max\{\alpha,1-\alpha\}$, and for
$1\leq\alpha\leq 2$ and $z=1$, and for $1\leq\alpha<\infty$ and $z=\alpha$.
See \cite{Audenaert} for these cases. In this paper Audenaert  and Datta
conjecture that monotonicity  holds for $1\leq\alpha\leq 2$ and
$\alpha/2\leq z<\alpha$, and for $2\leq\alpha<\infty$ and $\alpha-1\leq z<
\alpha$. Our contribution here is to prove their conjecture for
$1<\alpha=2z\leq 2$.

\begin{cl}
Let $\alpha=2z\in (1,2]$ and let $\rho,\sigma\in\Pn$. Then for any completely positive, trace preserving map $\mathcal E$ on $\Pn$,
$$
D_{\alpha,\alpha/2}(\rho||\sigma) \geq D_{\alpha,\alpha/2}(\mathcal E(\rho)||\mathcal E(\sigma)) \,.
$$
\end{cl}

\begin{proof}
By a classical argument due to Lindblad and Uhlmann, see, e.g.,
\cite{Carlen09,Frank13}, the monotonicity follows once it is
shown that
$$
(\rho,\sigma)\mapsto \tr\left( \sigma^{(1-\alpha)/\alpha} \rho^2 \sigma^{(1-\alpha)/\alpha}\right)^{\alpha/2} = \Phi_{2,2(1-\alpha)/\alpha,\alpha/2}(\sigma,\rho)
$$
is jointly convex. For $\alpha\in(1,2]$ this convexity follows from Theorem \ref{p2}.
\end{proof}


\noindent{\bf Acknowledgements} We thank Marius Lemm and Mark Wilde, as well as the anonymous referee, for useful remarks.

\end{document}